\documentclass{sig-alternate-per}
\usepackage{bm}

\usepackage[utf8]{inputenc}
\usepackage{amsmath, amssymb,MnSymbol }
\usepackage{array}
\usepackage[T1]{fontenc}
\usepackage{amssymb,amsfonts}
\usepackage{mathrsfs, relsize}

\usepackage{mathtools}
\newcommand{\RR}{\mathbf{R}}

\usepackage{algorithm}
\usepackage{algpseudocode}
\usepackage{graphicx}
\usepackage[font=small]{caption}
\usepackage[numbers]{natbib}
\usepackage[dvipsnames]{xcolor}
\usepackage{algorithm}
\renewcommand{\phi}{\varphi}

\newtheorem{theorem}{Theorem}
\newtheorem{proposition}{Proposition}
\newtheorem{lemma}{Lemma}
\newtheorem{definition}{Definition}

\newcommand{\cardR}{{|\setR|}}
\renewcommand{\geq}{\geqslant}
\renewcommand{\leq}{\leqslant}

\newcommand{\anti}[1]{\overline{#1}}

\newcommand{\argmin}{\mathrm{argmin}}
\newcommand{\setR}{{R}}
\newcommand{\setJ}{J}
\newcommand{\setP}{P}
\newcommand{\setC}{C}
\renewcommand{\c}{\bm{c}}
\newcommand{\y}{\bm{y}}
\newcommand{\z}{\bm{z}}
\newcommand{\p}{\bm{p}}
\newcommand{\x}{\bm{x}}
\newcommand{\A}{\mathbf{A}}
\newcommand{\w}{\bm{w}}
\newcommand{\m}{\bm{m}}
\newcommand{\dd}{\bm{d}}
\renewcommand{\b}{\bm{b}}
\newcommand{\0}{\bm{0}}
\renewcommand{\citet}[1]{[\citeauthor{#1}]}
\title{Lower Bounds for the Fair Resource Allocation Problem}

\numberofauthors{1}
\author{
\alignauthor
Zaid Allybokus$^{*\dag}$, Konstantin Avrachenkov$^*$, Jérémie Leguay$^\dag$ and Lorenzo Maggi$^\dag$\\
       \affaddr{$^*$ INRIA Sophia Antipolis}\\
       \affaddr{$^\dag$ Huawei Technologies, France Research Center}\\
       \email{\{zaid.allybokus, jeremie.leguay, lorenzo.maggi \}@huawei.com konstantin.avratchenkov@inria.fr}
}

\begin{document}
\conferenceinfo{IFIP WG 7.3 Performance 2017.}{Nov.\ 14-16, 2017, New 
 York, NY USA}

\maketitle
\abstract{The $\alpha$-fair resource allocation problem has received remarkable attention and has been studied in numerous application fields. Several algorithms have been proposed in the context of $\alpha$-fair resource sharing to distributively compute its value. 
However, little work has been done on its structural properties. In this work, we present a lower bound for the optimal solution of the weighted $\alpha$-fair resource allocation problem and compare it with existing propositions in the literature. Our derivations rely on a localization property verified by optimization problems with separable objective that permit one to better exploit their local structures. We give a local version of the well-known midpoint domination axiom used to axiomatically build the Nash Bargaining Solution (or proportionally fair resource allocation problem).
Moreover, we show how our lower bound can improve the performances of a distributed algorithm based on the Alternating Directions Method of Multipliers (ADMM). The evaluation of the algorithm shows that our lower bound can considerably reduce its convergence time up to two orders of magnitude compared to when the bound is not used at all or is simply looser.}


\section{Introduction}

The $\alpha$-fair resource sharing model, first studied in \citet{mo2000fair}, has already been investigated in numerous application domains, as well as its weighted variants. The weighted $(\w,\alpha)$-fair resource allocation problem is to find a vector $\x^* \in \RR_+^n$ such that 1) the utility 
\[ f^\alpha(\w,\x) =
\left\{ \begin{array}{ll}
\sum_{i=1}^n w_i\frac{x_i^{1-\alpha}}{1-\alpha}, & \alpha \neq 1,\\
 \sum_{i = 1}^n w_i\log(x_i), & \alpha=1\end{array}\right.
 \] is maximized at $\x = \x^*$, and 2) $\x^*$ lies in a feasible set defined by linear constraints of the form $\A\x \leq \c$ where $\c \in \RR_+^p$ is a capacity vector for a number $p$ of resources and $\A$ is the binary user-resource incidence $(p,n)$-constraint matrix, for a number $n$ of users, weighted by a positive vector $\w \in \RR_+^n$. 
The family of $(\w,\alpha)$-fair metrics is general and includes popular fairness concepts such as \emph{max-throughput} ($\alpha = 0$), \emph{proportional fairness}, also called \emph{Nash Bargaining Solution} ($\alpha = 1$), \emph{min-delay} ($\alpha = 2$) and arbitrarily close approximations of \emph{max-min fairness}  ($\alpha\to \infty$).
   
In this paper, we study the general weighted $(\w,\alpha)$-fair resource allocation problem under linear constraints and we propose a novel lower bound on its optimal solution. A lower bound is a positive vector $\bm{d} \in \RR_{+}^n$ respecting feasibility (that is, $\A\bm{d}\leq \c$) and such that $\x^* \geq \bm{d}$. Finding a lower bound in the context of fair resource sharing is of great interest -- it permits one to automatically define a minimal share that is attributable to each resource user as initialization of any exact computation that could take time, and may be helpful in the phase of design of a system. 
We seek to derive \emph{user-centric} formulas in the sense that their value for a specific user would depend only on the resources within a localized problem (and not on the global topology) and only on the users \emph{that compete over the same resources}. We then evaluate the formulas under different instance regimes and compare them to the literature in order to appreciate the improvements they provide.  

Remarkably, we also show how our lower bound can enhance the performance of a distributed algorithm based on the Alternating Directions Method of Multipliers (ADMM) (see  \citet{boyd2011distributed}) that can be invoked to solve optimally the $\alpha$-fair resource allocation problem. The ADMM is well-known for its fast convergence properties to modest accuracy; however, its performance is highly conditioned by the initialization of its so-called \emph{penalty parameter} that can, when badly tuned, induce an extremely poor convergence rate. Thus, tuning correctly the penalty parameter is a task that one should not neglect when using the ADMM. In light of recent studies (in particular, we exploit the results proven in \citet{deng2016global}), we demonstrate how our lower bound permits one to accomplish this task for our particular problem.
 
A well known lower bound of the proportionally fair ($\alpha = 1$) resource allocation was brought in as a building block of the axiomatization of the Nash Bargaining Solution and is commonly referred to as the \emph{midpoint domination} axiom \citet{deClippel2007}. It states that each user $i$ is given at least a fraction $\frac{w_i}{\sum_{j = 1}^n{w_j}}$ of their \emph{dictatorial allocation}, that is, the resource they would receive if the other users accept to receive $0$. We refer to the bound given by the midpoint domination axiom as the \emph{midpoint allocation}.
One can imagine that the midpoint allocation becomes arbitrarily poor as the total number of users becomes large, and its utility as a first estimation of the optimum allocation, negligible. Indeed, the formula includes the weights of the whole set of users and is independent of the problem's local structure. Similarly, the general lower bound found in \citet{marasevic2015fast} may suffer from these dependencies. 

Concerning proportional fairness ($\alpha = 1$), we give a more precise midpoint domination axiom, and provide a lower bound that we call \emph{local midpoint}.
Our lower bound on the proportionally fair allocation can be interpreted as a particular case of the midpoint domination axiom to locality -- now, each user $i$ is proportionally fairly attributable at least a fraction $\frac{w_i}{\sum_{j \in S_i}{w_j}}$ of their dictatorial allocation, where $S_i$ is not the total set of users, but \emph{the set of users in competition with the user $i$ for some resource}. Few works attempted at providing lower bounds for the general $(\w,\alpha)$-fair resource allocation. In fact, the most recent available bound is shown by \citet{marasevic2015fast}, and used by the authors for an initialization of their $\alpha$-fair heuristic. To the best of our knowledge, this is the best bound that could be found in the literature for the $\alpha$-fair resource allocation problem and we refer to it as the State-of-the-Art (SoA). 


The remainder of the paper is organized as follows: Section~\ref{sec:definition} is dedicated to the model definition and problem statement. Our lower bound presentation is addressed in Section~\ref{sec:lowerbound}. In Section~\ref{sec:admm}, we broadly remind the key features of the ADMM-based $\alpha$-fair distributed algorithm used for our illustration. The performance of the latter is shown in Section~\ref{sec:execution} and finally, Section~\ref{sec:conclusion} concludes the paper.

\section{Model Definition}
\label{sec:definition}


Let us start by formalizing the weighted $\alpha$-fair resource allocation problem. In this work, we adopt the terminology of rate control in fixed communication networks. Thus, a resource will be referred to as a \emph{link} and a user will be called a \emph{connection request} (or shortly, \emph{request}) from a source node to a destination node over a \emph{route} formed of several links. 

Let $\setJ$ be the set of network links, each link $j$ having a capacity $c_j\in \RR_+$.
Let $\setR$ be the set of  requests. Each request $r$ has a predefined route that identifies with a subset $\setJ_r \subset \setJ$ of links of the network. In turn, for each link $j\in \setJ$, $\setR_j \coloneqq \{r\in \setR; j\in \setJ_r\}$ is the set of all requests having a route that contains the link $j$. We define the link-route incidence $|\setJ|\times |\setR|$-matrix $\A$ as:

\[\A_{jr} =\left\{\begin{array}{lr} 1 & \mbox{ if } j\in \setJ_r \\ 0 & \mbox{ otherwise}\end{array} \right. \]
For each  request $r$, $x_r$ denotes the bandwidth allocated to $r$ along its route $\setJ_r$. We say that an allocation $\x = (x_r)_{r\in \setR}$ belongs to the feasibility set $\setC$ (or is \emph{feasible}) if it satisfies the capacity constraint~\eqref{eq:capacity} below:
\begin{equation}
\x \in \setC \Leftrightarrow\A\x \leq \c, \ \x\geq 0
\label{eq:capacity}
\end{equation}
where $\c = (c_j)_{j\in \setJ}$.
Each request $r$ is associated with a weight $w_r \in \RR_+$. The weight vector $ \w =(w_r)_{r\in \setR}$  accounts for a degree of relative importance of each request that can be defined at the discretion of the network. 
Weighted $\alpha$-fairness is formalized as in Definition~\ref{defalphafair} below.

\begin{definition}[$(\w,\alpha)$-fairness]
\label{defalphafair}
Let $ \setC\subset \RR_+^\cardR$ be a feasibility set defined as in \eqref{eq:capacity}, being a strict superset of $\{0\}$.
Let $\w \in \RR_+^\cardR$ and $\x^*\in \setC$. We say that $\x^*$ is $(\w,\alpha)$-\emph{fair} (or simply $\alpha$-fair when there is no confusion on $\w$) if the following holds:
\[\forall r\in \setR,\quad x^*_r >0\quad\mbox{and}\quad\forall \x\in \setC 
,\quad \sum_{r \in \setR}w_r\frac{x_r - x^*_r}{x^{*\alpha}_r} \leq 0. \]
Equivalently, $\x^*$ is $(\w,\alpha)$-fair if, and only if $\x^*$ maximizes the $\alpha$-fair utility function $f^\alpha$ defined over $\setC-\{\0\}$:
\begin{equation} \label{eq:centralProblem}
\max_{\x \in C} f^{\alpha}(\w,\x) = \sum_{r \in \setR} f^\alpha_r(w_r,x_r), \tag{$\setP_\alpha$}
\end{equation}
where
\[
f_r^\alpha(w_r,x_r) =
\left\{ \begin{array}{ll}
w_r \frac{x_r^{1-\alpha}}{1-\alpha}, & \alpha \neq 1,\\
w_r \log(x_r), & \alpha=1.
\end{array}\right.
\]

\end{definition}

%
%
%
  

 \section{Alpha-fairness --  a lower bound}
\label{sec:lowerbound}

In this section, we derive an explicit lower bound on the general $(\w,\alpha)$-fair resource allocation problem. Our lower bound only depends on the weight vector $\w$, the capacity vector $\c$ and the link-route incidence matrix $\A$. Moreover, the bound exploits the local structure of the problem, which prevents it from deteriorating systematically with the problem size. We compare it to the SoA bound that one can formulate as follows:

\begin{proposition}[\citet{marasevic2015fast}]
\label{marasevicbound}
Let the vector $\x^* $ be the optimal solution to the $\alpha$-fair resource allocation problem. Then, for all $r\in R$:
\begin{itemize}
\item if $0<\alpha \leq 1,\quad x_r^* \geq m_r(\alpha) \coloneqq \left( \frac{w_r}{w_{\max} M}\displaystyle \min_{j \in J_r} \frac{c_j}{|R_j|}\right)^{1/\alpha} c_{\max}^{1-1/\alpha}$
\item if $\alpha >1,\quad x_r^* \geq m_r(\alpha)\coloneqq \left( \frac{w_r}{w_{\max} M}\right)^{1/\alpha}\displaystyle \min_{j\in J_r} \frac{c_j}{|R_j|}\left(\frac{c_{\min}}{c_{\max}}\right)^{1-1/\alpha}$
\end{itemize}
where $w_{\max} = \max w_r$, $M = \min\{|R|,|J|\}$, $c_{\min} = \min c_j$ and $c_{\max} = \max c_j$.
\end{proposition}

We seek to improve the above bound by removing the global dependencies on $w_{\max}$, $|R_j|$, and $M$, $c_{\min}$ and $c_{\max}$, those parameters being the major degradation factor when the size or congestion of the problems increase. 

For each request $r\in \setR$, let $b_r \coloneqq \min_{j\in \setJ_r} c_j$. The so-called \emph{utopia point} $\b\coloneqq (b_r)_{r\in \setR}$ is the (infeasible when the problem is non trivial) allocation representing the value each request would receive if they were alone in the network, that is, its dictatorial allocation. Our bound for the $(\w,\alpha)$-fair allocation only depends on the utopia point (hence on the capacity vector $\c$), the matrix $\A$ and on the weight vector $\w$. For $r \in \setR$, let $\setR^r \coloneqq \{s\in \setR; \setJ_r \cap \setJ_s \neq   \varnothing \}$, i.e., the set of requests sharing at least one resource with $r$ and $\anti{\setR^{r}} \coloneqq \setR \minus \setR^{r}$. 

First of all, we use the separability of the objective function of Problem~\eqref{eq:centralProblem} to better estimate our lower bound on a restricted problem. Specifically, we prove a \emph{restriction lemma} (see Lemma~\ref{lemma}) that permits one to avoid unnecessary dependencies between requests that do not share resources together. Then, we prove our general lower bound on the corresponding restricted problems. Thanks to Lemma~\ref{lemma}, the bound remains unchanged in the original problem.

\subsection{A restriction lemma}

In this paragraph, we show that instead of evaluating our bound on Problem~\eqref{eq:centralProblem}, one can use a smaller request-centric problem. Specifically, let $\x^*$ denote the optimal solution of Problem~\eqref{eq:centralProblem} and let $r_0 \in R$ be \emph{an arbitrary request}. We define the \emph{restricted problem} at $r_0$, as the following:

 \begin{equation}
 \label{eq:restricted} \tag{$\widetilde{\setP_{r_0}}$}
 \min \sum_{r\in \setR^{r_0}} -f_r^\alpha(w_r,x_r) 
 \end{equation}
 \[\mbox{s.t. } \sum_{r  \in \setR_j \cap \setR^{r_0}} x_r \leq \tilde{c}_j \coloneqq c_j \quad \forall j\in \setJ_{r_0}\]
 \[\mbox{and } \sum_{r\in \setR_j \cap \setR^{r_0 }} x_r \leq \tilde{c}_j \coloneqq c_j - \sum_{r\in \setR_j \cap \anti{\setR^{r_0}}} x_r^* \quad \forall j\in \setJ - \setJ_{r_0}.\]
 Intuitively, Problem~\eqref{eq:restricted} arises when the allocations of all the requests that do not share any link with $r_0$ are fixed to their optimal $\alpha$-fair value (that is, following the vector $\x^*$), and one needs to compute the $\alpha$-fair allocation of the remaining requests, that is, the requests within $R^{r_0}$ that share at least one resource with $r_0$. The capacity constraints are thus updated taking into account the amounts of resources that are already allocated, as shows the second line of the constraints. Note in passing that all the links in $J-J_{r_0}$ that do not serve any of the requests within $R^{r_0}$ form trivial constraints in \eqref{eq:restricted} and can hence be removed without any loss.
 
 We then have the following result:
 
 \begin{lemma}
 \label{lemma}
The restriction to~\eqref{eq:restricted}  does not change the optimal allocation of the remaining requests: if $\x$ is the optimal solution of the Problem~\eqref{eq:restricted}, then, $x_s = x^*_s$, for $s\in \setR^{r_0}$.
\end{lemma} 

\begin{proof}
Consider the problem:
 \begin{equation}
 \label{restricted}
 \min \sum_{r\in \setR} -f_r^\alpha(w_r,x_r)
 \end{equation}
 \[\mbox{s.t. } \A\x\leq \c\]
 \[ x_r \geq x_r^* \quad  \forall r \in \anti{\setR^{r_0}}. \] 
It suffices to show that the problems~\eqref{restricted} and \eqref{eq:restricted} are equivalent. Then, the unicity of the solutions permits one to conclude.

We know that the problem~\eqref{restricted} is feasible, as $\x^*$ is a feasible point. Denote its optimal solution by $\tilde{\x}$. We remark that both $\x^*$ and $\tilde{\x}$ are feasible for both problems $(\setP_\alpha)$ and \eqref{restricted}. Hence, by optimality, we necessarily have $f^\alpha(\w,\tilde{\x}) = f^\alpha(\w,\x^*)$. Moreover, for instance, problem~$(\setP_\alpha)$ has a unique optimal solution. Thus, $$\x^* = \tilde{\x}.$$  Particularly for $r\in \anti{\setR^{r_0}}, x^*_r = \tilde{x}_r.$ Thus, we can fix the values $x_r = x_r^*$ for $r\in \anti{\setR^{r_0}}$ without changing the optimum. Thus, Problem~\eqref{restricted} is equivalent to the restricted problem~\eqref{eq:restricted}.
\end{proof}

Thanks to Lemma~\ref{lemma}, we are now ready to present our lower bound on the $\alpha$-fair allocation based on the structure of the restricted problems. 

\subsection{Lower bound}

We now show the main result of this paper. We define the \emph{local midpoint} $\p$ as the following:  $$\forall r\in R\quad p_r \coloneqq \frac{w_r}{\displaystyle\sum_{s\in R^r} w_s}b_r.$$ 

%
%

\begin{theorem}
\label{theorem}
Let $\x^*$ denote the optimal solution of problem \eqref{eq:centralProblem}. Let $r_0 \coloneqq \argmin_{s\in R} p_s$. Then, $\x^*$ can be lower bounded as follows:
 \begin{itemize}
 \item if $\displaystyle\alpha \geq 1, \quad \forall r\in \setR \quad x_r^* \geq d_r(\alpha) \coloneqq p_{r_0}^{1-1/\alpha}p_r^{1/\alpha}$ 
 \item if $\displaystyle 0\less \alpha\leq 1,\quad \forall r\in \setR\quad x_r^* \geq d_r(\alpha) \coloneqq\left( \frac{w_r b_r}{\displaystyle \sum_{s\in \setR^r}w_s b_s^{1-\alpha}}\right)^{1/\alpha}$.
 \end{itemize}
 \end{theorem}

\begin{proof}
We first prove the proposition for $\alpha \geq1$. Let us define the request $r_{\min}$ as the request with the least optimal allocation: $r_{\min} = \argmin_{s\in R} x^*_s.$
By definition of $r_0$, we have: \begin{equation}p_{r_{\min}} \geq p_{r_0} \end{equation}
Let $r\in \setR^r$. By Lemma~\ref{lemma}, it suffices to show the inequality in the restricted problem $(\widetilde{\setP_{r}})$ associated to $r$. Let $\setC^r$ denote its feasible set. Thus, for all $(x_s)_{s\in \setR^{r}} \in \setC^r$ we have:

\[\sum_{s\in\setR^r} w_s \frac{x_s - x^*_s}{x^{*\alpha}_s} \leq 0,\]
This inequality holds for all feasible $(x_s)_{s\in R^r} \in C^r$. Thus, we evaluate it at the dictatorial allocation of $r$, that is, at the point $\x$ defined as $x_r = \tilde{b}_r$ and $x_s = 0$ for all $s\in R^r - \{r\}$.  Let us note in passing that $\tilde{b}_r = \min_{j\in \setJ_r} \tilde{c}_j = \min_{j\in \setJ_r} c_j = b_r$.
 Thus,  
 \[w_r \tilde{b}_r =w_r{b_r} \leq x_r^{*\alpha}\sum_{s\in \setR^r} w_s x_s^{*1-\alpha} \leq \left(\sum_{s\in \setR^r} w_s\right) x_{r_{\min}}^{*1-\alpha} x_r^{*\alpha},\]
 where we remind that  $r_{\min} = \argmin_{s\in R} x^*_s$ and $1-\alpha \leq0$. Rearranging the terms, one gets:
 \[\frac{w_r b_r}{ \displaystyle\sum_{s\in \setR^r}w_s} x_{r_{\min}}^{*\alpha -1}\leq x_r^{*\alpha},\]
 which yields:
 \begin{equation}\label{proof:general}
p_r^{1/\alpha} x_{r_{\min}}^{*1 -1/\alpha}\leq x^*_r
\end{equation}
In particular, applying equation~\eqref{proof:general} to $r = r_{\min}$, we get:

\begin{equation}\label{proof:min}
x^*_ {r_{\min}} \geq p_{r_{\min}} \geq p_{r_0}
\end{equation}
Finally, we plug equation~\eqref{proof:min} in equation~\eqref{proof:general} to obtain the desired lower bound on $x^*_r$ (because $1-1/\alpha \geq 0$).

Next, we show the bound for $0 \less \alpha < 1$. In the same fashion, we look at the restricted problem. Let $r\in \setR$ and consider its restricted problem. Then, one has:

\[\frac{w_r{b_r}}{ x_r^{*\alpha}} \leq\sum_{s\in \setR^r} w_s x_s^{*1-\alpha} 
\leq \sum_{s\in \setR^r} w_s \tilde{b}_s^{1-\alpha}\leq \sum_{s\in \setR^r} w_s b_s^{ 1-\alpha}. \]
Rearranging the terms finally provides the desired bound. For any value of $\alpha$, one can remark that the bound $(d_r(\alpha))_{r\in R}$ only depends on the capacity vector $\c$, the weight vector $\w$, and the link-route incidence matrix $\A$.
\end{proof}
 \begin{figure*}[t]
 \includegraphics[width = \textwidth]{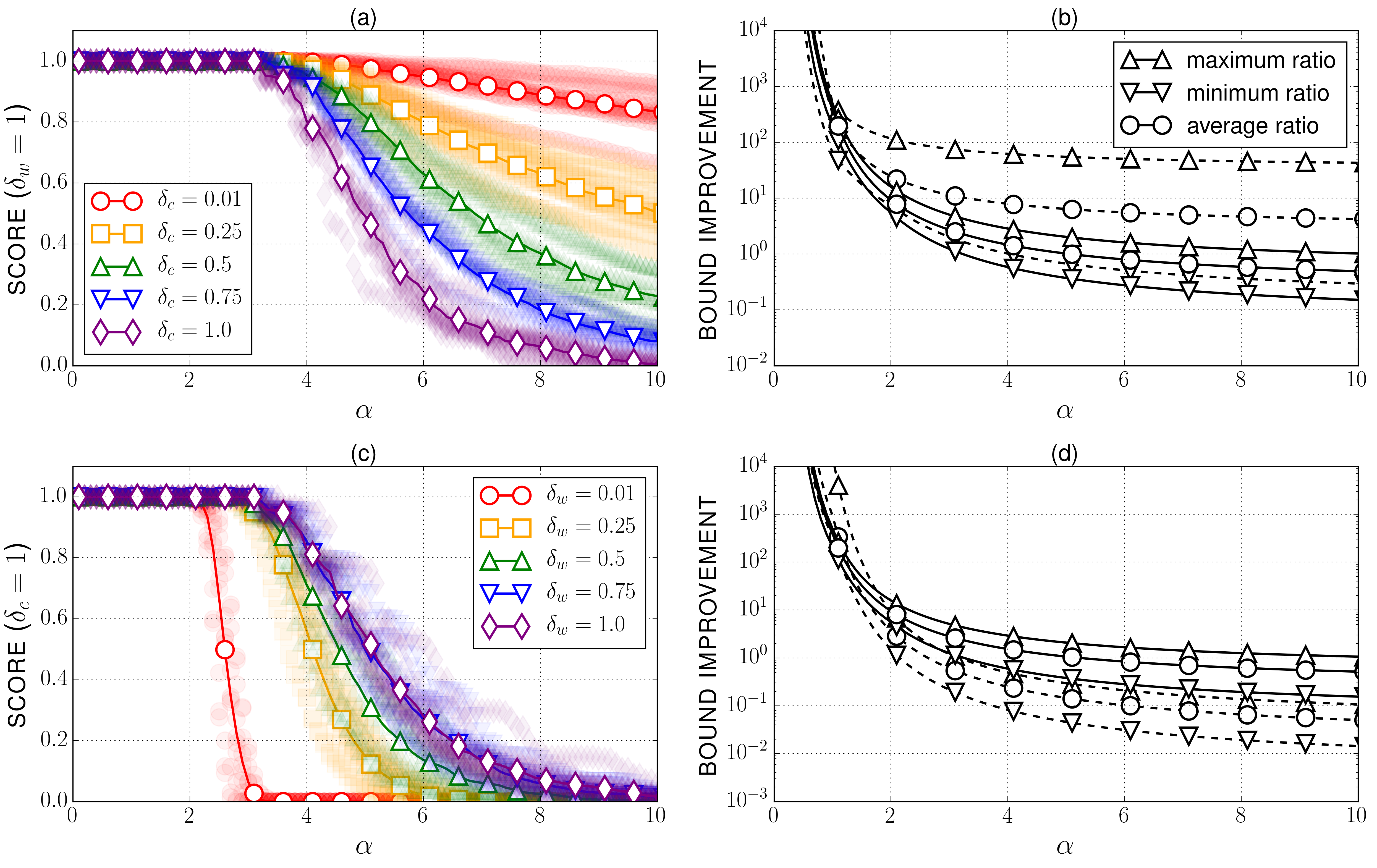}
 \caption{A comparison of the two bounds. The scores, and the minimum, average and maximum bound improvements are illustrated in the cases of (a)-(b) a constant $\delta_w$ for different values of $\delta_c$, and of (c)-(d) a constant $\delta_c$ for different values of $\delta_w$. Figures (b) and (d) show the bound improvements in the two extreme situations $\delta_c$ (resp. $\delta_w$) $=0.01$ (resp. $1$) in dashed lines (resp. solid lines).}
  \label{fig:1linkbound}
 \end{figure*}
 
 \subsection{Illustration}
 \label{illustration}
 To conclude this section, we illustrate a comparison of the two presented lower bounds $\m$ and $\dd$ introduced in Proposition~\ref{marasevicbound} and Theorem~\ref{theorem}, respectively,  under different regimes. Given the formulas, one can remark that the sensitivity of the bound to arbitrary problem sizes should be lessened as now more focused on local structures. For $\alpha  \leq 1$, we obtain request-centric formulas. For general $\alpha >1$, this elimination came  with the dependency on the global minimum local midpoint value $p_{r_0}$. Intuitively, one can remark that the two bounds may react differently to a fluctuating asymmetry of the weight vector $\w$ or the capacity vector $\c$, namely, a variation of the two parameters $\delta_w \coloneqq \frac{\min w_r}{\max w_r}$ and $\delta_c \coloneqq \frac{c_{\min}}{c_{\max}}$. For a better vision, we illustrated this behavior in Figure~\ref{fig:1linkbound}. 
 
 The two bounds were compared on instances with 1000 requests over a same graph of type \emph{barabasi(100,4)} (see \citet{barabasi}). The routes were generated at random by taking the shortest path between pairs of sources and destinations drawn uniformly at random. The weights  (resp. link capacities) were also drawn uniformly at random within intervals $I$ satisfying $\inf I / \sup I =\delta_w$ (resp. $\delta_c$). For each instance, and each $\alpha$, we define the \emph{score} of $\dd$ as the number $|\{r : d_r(\alpha) > m_r(\alpha)\}| / |R|$. The score represents the proportion of requests for which our bound  $\dd(\alpha)$ beats the SoA bound $\m(\alpha)$ for a particular $\alpha$. In Figure~\ref{fig:1linkbound}(a), the parameter $\delta_w$ was fixed to $1$ (which namely means $\w = \bm{1}$) and we plotted the score of $\dd$ versus $\alpha$ for different values of $\delta_c$. Figure~\ref{fig:1linkbound}(c) shows the score in the other extreme situation $\delta_c =1$ (which means all the link capacities are equal) for different values of $\delta_w$. 
 
 In order to appreciate the quality of the bound improvement, if any, we plotted, in Figures~\ref{fig:1linkbound}(b) and \ref{fig:1linkbound}(d), the corresponding bound improvements, measured with the values of the ratios $d_r(\alpha)/m_r(\alpha)$. To preserve readability of the plots, we represented only the extreme situations corresponding to the values $\delta_c =0.01$ (dashed lines) and $\delta_c = 1$ (solid lines) for Figure~\ref{fig:1linkbound}(b) and to the values $\delta_w =0.01$ (dashed lines) and $\delta_w = 1$ (solid lines) for Figure~\ref{fig:1linkbound}(d). Figures~\ref{fig:1linkbound}(b) and \ref{fig:1linkbound}(d) show the best, worst, and average improvements encountered in the same problem instance. All the points represented in Figure~\ref{fig:1linkbound} correspond to an average over 10 instances generated under identical conditions. In Figures~\ref{fig:1linkbound}(a) and \ref{fig:1linkbound}(c), we also included the specific points as translucent scattered markers.  

According to Figure~\ref{fig:1linkbound}, our bound is an absolute improvement for values of $\alpha$ in the interval [0,2] (thus including the max-throughput, proportional fairness, and min-delay popular concepts) in all situations. Particularly for proportional fairness, the simulations show that we improved the bound $\m$ by two orders of magnitude in all situations. For min-delay fairness, the bound is generally improved on average by a multiplicative factor between 1 and several tens. For greater values of $\alpha$, it is interesting to see that either $\dd$ or $\m$ is more adapted to certain problem structures. For instance, $\dd$ will be of greater interest when the network link capacities are more heterogeneous, $\delta_c \ll 1$ (which may correspond to situations where the network is asymmetrically congested), whereas $\m$ is more adapted to asymmetrically weighted problems, $\delta_w \ll 1$. One can thus conclude that the two available bounds complement each other for general $\alpha \geq 1$.

After presenting our bound, we now demonstrate how it permits one to boost the performance of an algorithm that solves the $(\w,\alpha)$-fair resource allocation problem.
 
%
 
 The next section is dedicated to the presentation of the algorithm, based on ADMM.
\section{Fast and Distributed ADMM (FD-ADMM)}
\label{sec:admm}

Several approaches may be used to tackle the $(\w,\alpha)$-fair resource allocation problem (e.g., see \citet{kelly1998rate} and \citet{palomar2006tutorial} for a tutorial). One of them is the Alternating Directions Method of Multipliers (ADMM)  (see, e.g., \citet{boyd2011distributed}). The ADMM is well known for its distributivity properties that permit one to decouple constraints handled in parallel then plugged in together by means of \emph{consensus constraints}. 
In \citet{allybokus2017real}, these properties are used to design a fully distributed algorithm that solves optimally the problem in the context of traffic rate control in distributed Software-Defined Networks.
For a description of the general ADMM framework, the reader may refer to \citet{boyd2011distributed}, and for a more detailed construction of the presented algorithm, to  \citet{allybokus2017real}. In this section, we briefly describe the design of the distributed algorithm used in the latter.

\subsection{Algorithm overview}

Assume the network links are split into a number $P \geq 1$ of domains. Each domain $p$ corresponds to a subset $J_p \subset J$ of links forming a covering of the whole set $J$:
\[
\displaystyle\bigcup_{p=1}^P J_p = J.
\]
 For $p = 1 \ldots P$, let $R_p = \{r\in R, J_p \cap J_r \neq  \varnothing\}$  be the set of requests that traverse domain $p$, and $I_r = \{q\in [1,P] ; r\in R_p\}$ the set of domains the request $r$ traverses. The problem~\eqref{eq:centralProblem} can thus be rewritten as:
 :
 \begin{align}
     	& \min_{\x\in \setC} \sum_{r\in R} -f_r^\alpha(w_r,x_r) \nonumber\\
  =	&  \min_{\x\in \setC} \sum_{p = 1}^P \sum_{r\in R_p} -\frac{1}{|I_r|} f_r^\alpha(w_r,x_r)\nonumber\\
  =	&  \min_{\x} \sum_{p = 1}^P  \left\{\iota_p(\x) + \sum_{r\in R_p} -\frac{1}{|I_r|} f_r^\alpha(w_r,x_r)\right\}  \nonumber\\
  :=	& \min_{\x} \sum_{p = 1}^P  \iota_p(\x)+g_p(\w,\x),\label{eq:gpsingle}
 \end{align}
 where $\iota_p$ is the indicator function of the capacity subset associated to domain $p$: \[\iota_p(\x) = \sum_{j\in J_p} \iota_j(\x), \quad \iota_j(\x) = \left\{\begin{array}{lr} 0 & \displaystyle\mbox{ if } \sum_{r \in R_j} x_r\leq c_j \\ \infty & \mbox{ otherwise}\end{array} \right.\]
Further, we separate artificially the problem by creating a private variable $\x_p \in \RR^{|R_p|}$ for each domain $p$, and by enforcing the agreement upon their values between domains with consensus constraints. Problem formulation~\eqref{eq:gpsingle} now reads:
 
 \begin{align}
\min 		& \quad\sum_{p=1}^P g_p(\w,\x_p) + \iota_p(\x_p)\nonumber \\
 \mbox{ s.t }	& \quad x_{pr} = x_{qr} \quad\forall p,q \in I_r\quad \forall r\in R\label{eq:maconsensus}\\
 		& \quad \x_p \in \RR^{|R_p|} \quad \forall p =1\ldots P.\nonumber
 \end{align}
 
Finally, we decompose the problem by separating the private objective of each domain. For each domain $p$, and each $j \in J_p$, the vector $\y_j$ defines a copy of the variable $\x_p$ for link $j$ and is reserved for the component function $\iota_j$. We can write Problem~\eqref{eq:maconsensus} in the following form:

 \begin{align}
\min 		& \quad\sum_{p=1}^P g_p(\w,\x_p) + \sum_{j\in J_p} \iota_j(\y_j) \coloneqq \sum_{p\in P} G_p(\w,\x_p,\y_p) \nonumber\\
 \mbox{ s.t }	& \quad x_{pr} = x_{qr} \quad\forall p,q \in I_r\quad \forall r\in R\label{eq:admm1}\\
 		& \quad x_{pr} = y_{jr} \quad \forall j\in J_p\quad \forall r\in R_j \nonumber \\
 		& \quad \x_p \in \RR^{|R_p|} \quad \forall p =1\ldots P\nonumber\\
 		& \quad \y_j \in \RR^{|R_j|} \quad\quad \forall j \in J\nonumber \\
 		& \quad \y_p = (\y_j)_{j\in J_p} \quad \forall p \in P\nonumber 
 \end{align}

\noindent Let $\chi$ denote the indicator function of the feasible set~\eqref{eq:admm1}. Then, the formulation takes the compact 2-block form:
\begin{align}
\min 		& \sum_{p\in P} G_p(\w,\x_p,\y_p) + \chi((\x'_p)_{p\in P},(\y'_p)_{p\in P})\label{eq:compact}\\
\mbox{s.t. } & (\x_p,\y_p) = (\x'_p,\y'_p)\nonumber
\end{align}

Applied to the last formulation~\eqref{eq:compact}, the distributed ADMM is described in Algorithm~\ref{FDADMM}. In lines~\ref{dualj} and \ref{updatecentral}, the variables $\bm{u}_j\in \RR^{|R_j|} $ and $\bm{v}_p \in \RR^{|R_p|}$ are dual variables associated with the constraints $\{\y_j = \y'_j\}$ and $\{\x_p = \x'_p\}$, respectively in \eqref{eq:compact}. Also, {\sc P}$(j,\cdot)$ is the Euclidean projection onto the simplex $\{\y_j \in \RR^{R_j} \mbox{ s.t. }  y_{jr} \geq d_r \mbox{ and }\sum_{r\in R_j} y_{jr} \leq c_j\}$, $\lambda \gtr 0$ is a scalar reciprocal penalty parameter, and $\bm{d} \in \RR^{R}$ is a lower bound on the $(\w,\alpha)$-fair solution that will be computed with the input parameters.

 \begin{algorithm}[t]
 \caption{Fast Distributed ADMM (FD-ADMM)}\label{FDADMM}
 \begin{algorithmic}[1]
 \Procedure{$\mbox{of\textsc{ Domain} $p$}$}{}
 
 \State \mbox{\sc Receive} $\tilde{\bm{z}}_p = (\tilde{z}_r)_{r\in R_p}$ 
 \State {\sc Receive updated reciprocal penalty } $\lambda_0$ {\sc from master}\label{receive} 
 
 \For{$j\in J_p$}
 	\State $ u_{jr} \gets u_{jr} +y_{jr} -\tilde{{z}}_{pr}\quad \forall r\in R_j$ \label{dualj}
 	\State $y_j \gets \mbox{\sc P}(j,\tilde{\z}_p - \bm{u}_j)$ \label{projectionstep}
\EndFor
\For{$r\in R_p$}
\State $v_{pr} \gets v_{pr} +x_{pr} - \tilde{z}_{pr}$
\State $ x_{pr} \gets \argmin_{x}\left\{ -f^\alpha_r(w_r,x) + \frac{1}{2\lambda_0}|| x - (\tilde{z}_{pr} - v_{pr})||^2\right\}$\label{updatecentral}
\EndFor
\EndProcedure

\Procedure{$\mbox{of \textsc{Master}}$}{}
\State {\sc Compute lower bound $\bm{d}$ and $\lambda_0$ using Eq.  \eqref{eq:optimalstep}}
\While{termination condition not met}
\State $\lambda_0 \gets \lambda$ updated using {RB, by \citet{he2000alternating}}
\For{$r\in R$}
\State \State $ \tilde{z}_{r} \gets \frac{1}{|J_r|+|Ir|} \left(\sum_{j\in J_r} y_{jr}+ \sum_{q\in I_r} x_{qr}\right)$\label{alg2:comm}
\EndFor
\For{$p\in P$}
 \State \mbox{\sc Send} $\tilde{\bm{z}}_p = (\tilde{z}_r)_{r\in R_p}$ \mbox{\sc to domain} $p$\label{send} 
\EndFor
\EndWhile
\EndProcedure
 \end{algorithmic}
 \end{algorithm}
\newpage
 \subsection{Performance}
 
 The convergence of ADMM is provably known since the 1990s (see \citet{Eckstein1992}), and its convergence rate has been widely studied. Today, the most general convergence rate of ADMM is known to be $O(1/T)$ ($T$ being the iteration count), and linear convergence rates are provably obtained for strongly convex problems. Nevertheless, the performance of the ADMM remains highly sensitive to the initialization and the update of the penalty parameter. 
In \citet{deng2016global}, the linear convergence rate of ADMM for strongly convex problems is quantified and optimized with regards to the penalty parameter, which yields an optimal tuning of it. 
Its value depends on the (global) strong convexity and the Lipschitz gradient moduli of the objective function, if those are finite. 
In \citet{allybokus2017real}, this result is applied to a central strongly convex equivalent formulation of our problem to derive an approximate adaptive tuning of the distributed version of the algorithm. The adaptive penalty parameter is computed as the optimal parameter of the centralized formulation, $\lambda_0$, given according to the formula
\vspace{-2mm}
 \begin{equation}\lambda_0 = \frac{1}{\sqrt{\sigma L}}, \label{eq:optimalstep} \end{equation}
 \vspace{-2mm}

\noindent where $\sigma$ is the strong convexity modulus of $f^\alpha(\w,\cdot)$ and $L$ is the Lipschitz modulus of its gradient. In fact, the fairness functions have singular values near $0$, which make the Lipschitz modulus not globally defined, unless the feasible set is reduced from below by means of a positive lower bound $\bm{d}$ of the optimal solution. Thus, Equation~\eqref{eq:optimalstep} is applied to $L = L_{\bm{d}}$ where $L_{\bm{d}}$ is the Lipschitz modulus of the gradient of the objective over the set of feasible points $\x$ verifying $\x\geq \bm{d}$.   

Adaptive penalty parameter schemes have been proposed to tackle this issue and provably bring consistent improvement of the convergence behavior of ADMM. One remarkable adaptive scheme can be found in \citet{he2000alternating}, in which the authors introduce the  \emph{residual balancing} (RB) principle which consists of shrinking or expanding the penalty parameter whenever the primal and dual residuals are unbalanced. 
For a definition of RB, we refer the interested reader to \citet{he2000alternating}. Although this scheme helps making the ADMM less dependent from initialization, empirical behaviors of the algorithm however suggest that there is still room and interest for better initialization. To demonstrate this, we adopt residual balancing as a default adaptive scheme of our penalties in all the algorithms of the present paper.  

In Section~\ref{sec:lowerbound}, we introduced the lower bound $\dd$ (Theorem~\ref{theorem}) on the $(\w,\alpha)$-fair optimal allocation. Next, we demonstrate how this bound permits one to enhance the performance of the ADMM for the $(\w,\alpha)$-fair resource allocation problem, and we compare it with the performance brought by the SoA bound $\m$ (Proposition~\ref{marasevicbound}). Although the lower bound permits one to adjust quickly a minimal individual resource allocation that would never be violated during the running time of the algorithm, the major feature of its introduction is in that it permits one to define an initialization of the penalty parameter that could enhance the algorithm performance. Indeed, the initialization can provide spectacular convergence acceleration, whereas reducing the feasible set at the projection line~\ref{projectionstep} of Algorithm~\ref{FDADMM} does not seem to matter, illustrating the fast convergence of FD-ADMM to modest accuracy. These observations are illustrated in the next section.

 \section{Execution}
  \label{sec:execution}
In the present simulations, we dedicate our performance evaluation to the proportionally fair resource allocation problems ($\alpha = 1$). In this section, we demonstrate the gains achievable with only tuning the initial penalty parameter of the FD-ADMM by comparing several initialization schemes. Indeed, the only difference between the different algorithms that we compare is in that the initial penalty parameter $\lambda_0$ is chosen either arbitrarily -- FD-ADMM($\lambda_0 = \lambda$), or according to Equation~\eqref{eq:optimalstep} applied to the bound $\m$ (FD-ADMM-MB) or $\dd$ (FD-ADMM-LB).
 
The problem instances were generated under the same conditions as in Section~\ref{illustration}. As it appears the parameter $\delta_w$ can deteriorate importantly the quality of our bound when small, we execute the simulations under two different situations 1) $w_r \in [.9,1]$, and 2) $w_r \in [.1,1]$.

\subsection*{Performance results}
 \begin{figure}[t]
 \vspace{2mm}
 \hspace{-5mm}
 \includegraphics[width =.5 \textwidth]{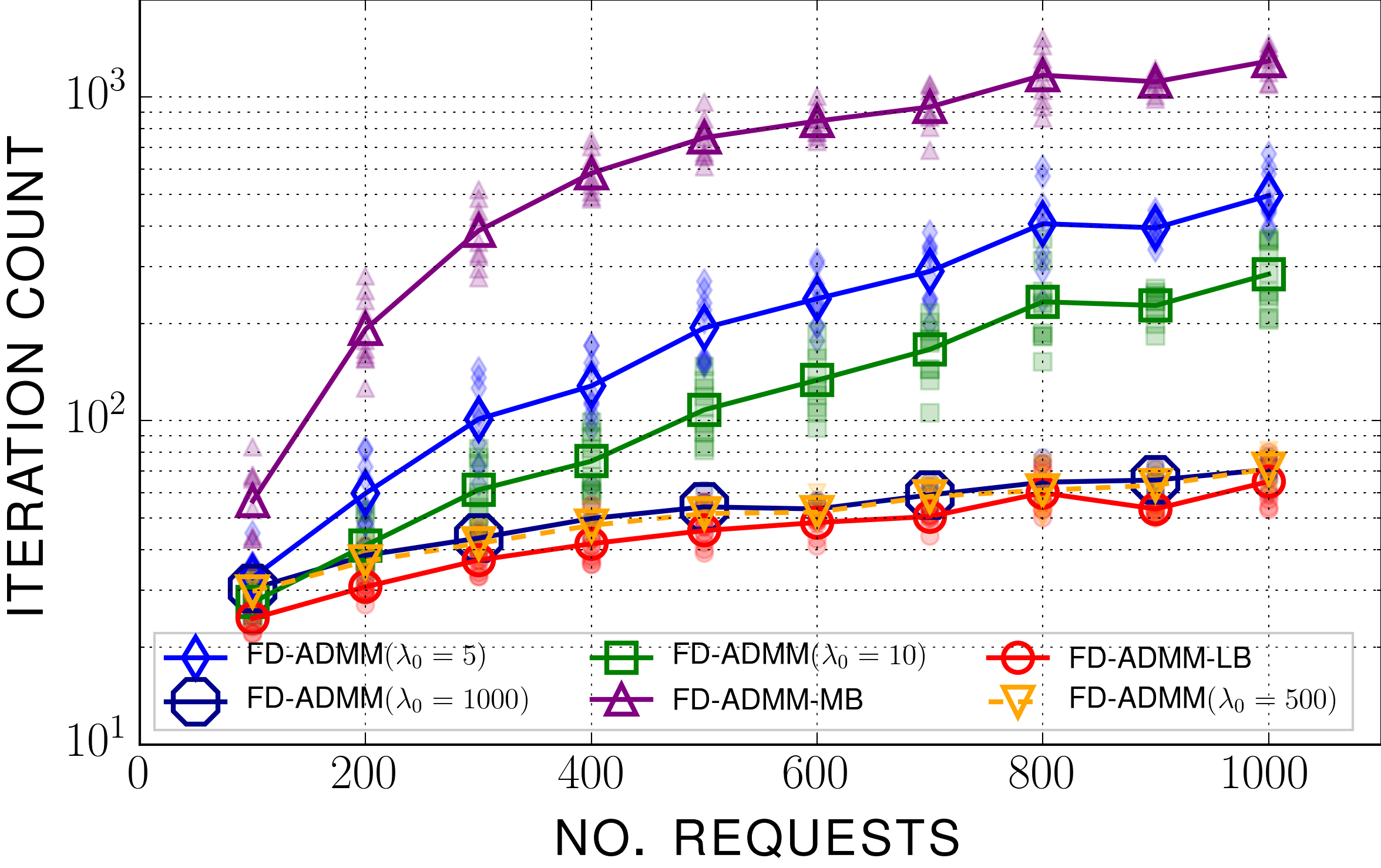}
 \caption{Iteration count versus the number of connection requests in situation 1. For FD-ADMM-LB and FD-ADMM-MB, the reciprocal penalty initialized value $\lambda_0$ lies in [110-150] and in [1, 6], respectively. }
  \label{fig:1}
 \end{figure}
In Figures~\ref{fig:1} and \ref{fig:2}, we plotted the iteration count of the algorithms under situations 1 and 2, respectively. The algorithms stop when the primal and dual residuals of the ADMM algorithm (see, e.g., \citet{boyd2011distributed}) fall below $10^{-2}$ (relatively modest accuracy). For each problem size in terms of number of different requests, we generated 10 instances of the corresponding size randomly and plotted the average performance. The specific points are also represented with scattered translucent markers to account for the exact performance of each algorithm. For each situation, we observed the performances of FD-ADMM-LB, in particular, its average initial reciprocal penalty parameter given by Equation~\eqref{eq:optimalstep} and chose several initialization values below and above this average to account for the effect of this initialization on the algorithm's performance.
  \begin{figure}[t]
\vspace{1mm}
\hspace{-5mm}
 \includegraphics[width =.5 \textwidth]{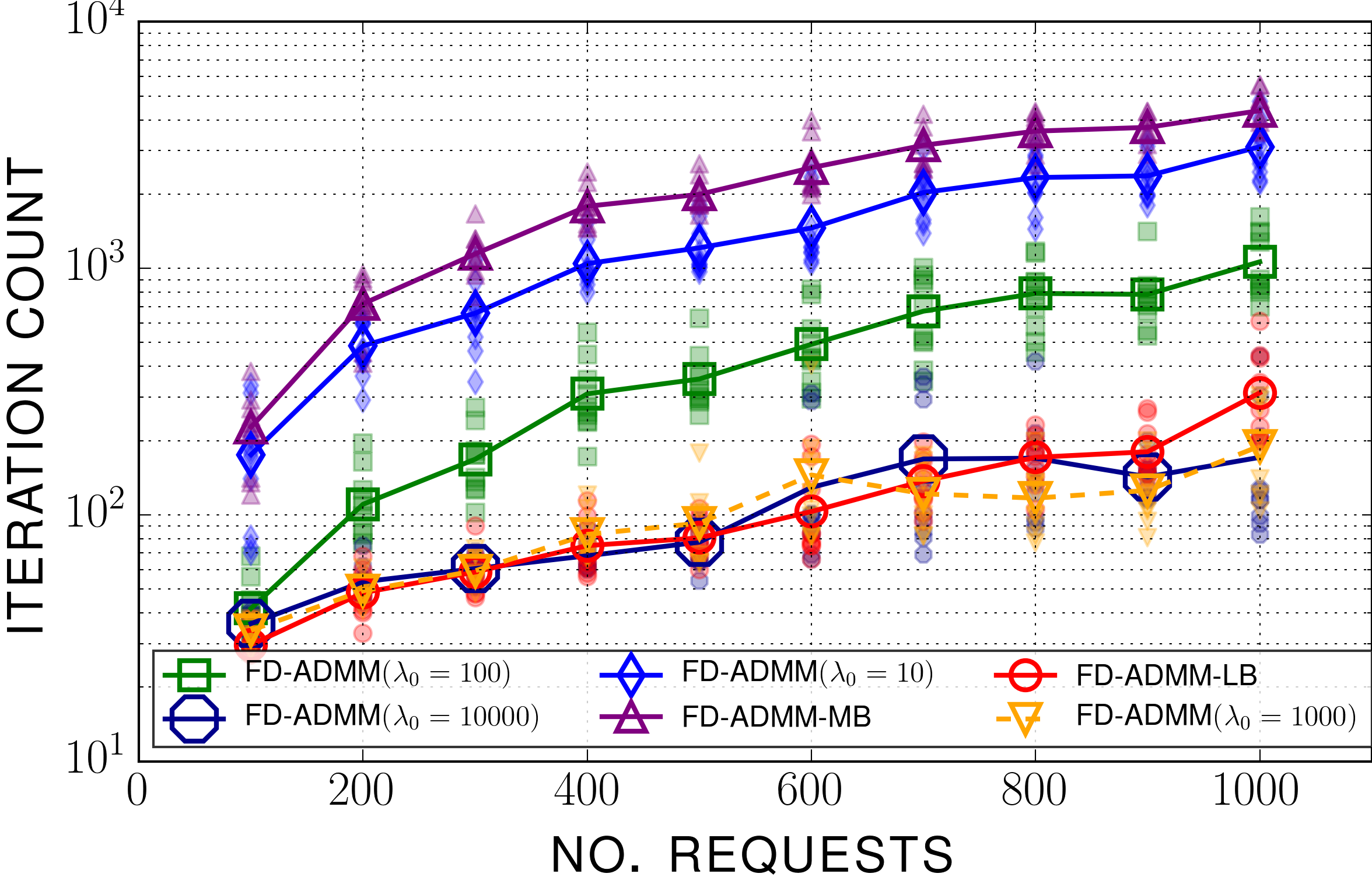}
 \caption{Iteration count versus the number of connection requests in situation 2. For FD-ADMM-LB and FD-ADMM-MB, the reciprocal penalty initialized value $\lambda_0$ lies in [500,1000] and in [0.1, 2.5], respectively. }
  \label{fig:2}
 \end{figure}
 
\textbf{Situation 1 (Figure~\ref{fig:1}).} We observe a spectacular improvement of the FD-ADMM algorithm from the scheme FD-ADMM-MB to the scheme FD-ADMM-LB, corresponding to a reduction of the running iteration count of two orders of magnitude. When $\lambda_0$ is chosen larger than the one for FD-ADMM-LB, although the performances seem satisfactory, one can observe that FD-ADMM-LB still executes faster on average.

\textbf{Situation 2 (Figure~\ref{fig:2}).} The same improvement of the performances, related to the introduction of our lower bound, is observed. It seems that for lower initialization value of $\lambda_0$, the algorithms demonstrate poorer performances. Nevertheless, one can observe that higher values of $\lambda_0$ can provide algorithms with, although not consistently, better performances than FD-ADMM-LB on average. Although this phenomenon can seem surprising after a look at situation 1, one can explain it with the fact that when the vector $\w$ is highly unbalanced (as it is the case when its values are uniformly drawn at random within $[.1,1]$) the objective function $f^\alpha(\w,\cdot)$ obtains highly asymmetric structure. Indeed, the computation of the strong convexity modulus of $f^\alpha$ in \citet{allybokus2017real} in order to obtain a desirable initialization $\lambda_0$, shows that the factor $\sigma$ in Equation~\eqref{eq:optimalstep} in fact corresponds to the smaller strong convexity modulus of the functions $f^\alpha_r(w_r,\cdot)$, which is proportional to $w_r$. Not surprisingly then, this evaluation becomes poorer when the vector $\w$ becomes unbalanced. 
Thus, it is worth considering that an accurate penalty parameter tuning becomes more difficult when the weighted fairness function symmetry is poor. 
Nevertheless, our simulations suggest that initializing the reciprocal penalty parameter according to Equation~\eqref{eq:optimalstep} applied to our lower bound permits one to obtain a satisfactory performance of the FD-ADMM. We believe this scheme can be improved in order to tackle a potential performance issue under highly asymmetric realizations of the $(\w, \alpha)$-fair resource allocation problem characterized by a very low value of $\delta_w$.


 \section{Conclusion}
 \label{sec:conclusion}
  
 We studied the structure of weighted $(\w,\alpha)$-fair allocation problems and proposed a lower bound that permits one to better understand the problem's features. The $(\w,\alpha)$-fair allocation can be lower bounded individually and locally (that is, each user, or request, has a minimal guarantied allocation that depends on its individual weight and that of a locally reduced subset of users). 
 We compared experimentally our bound with the best bound available in the literature, and showed that we can provide consistent improvement in the case of high asymmetry of the capacity vector $\c$ (which may describe congested networks situations) or in the case of a suitable symmetry of the fairness measures (which may cover situations where the requests have  balanced relative priorities). We believe that the bound derived in the present paper for general fairness concepts ($\alpha >1$) can be further improved, and intend to soften its dependencies on the global minimum local midpoint value $p_{r_0}$. Our intuition suggests this would improve considerably the quality of our general bound.
To demonstrate the utility of our derivation, we showed as an illustration how the introduction of this lower bound can remarkably improve the performances of an iterative distributed algorithm, the FD-ADMM, that solves the problem optimally, by a simple initialization of a penalty parameter. We also observed that the initialization scheme allows a remarkably satisfactory tuning of the FD-ADMM, and that this accuracy may impoverish as the asymmetry of the weighted problem strengthens. In the future, we envision to study this situation and strengthen our bound in order to possibly empower the initialization scheme, providing more robustness to the technique with respect to asymmetry.

%
 \small
\bibliography{biblio}
\bibliographystyle{ACM-Reference-Format-Journals}

\end{document}